\documentclass[11pt]{article}
\usepackage{graphicx} 
\usepackage{xcolor}
\usepackage{color, colortbl}
\definecolor{Graydark}{gray}{0.75}
\definecolor{Graylight}{gray}{0.90}

\usepackage[
backend=biber,
style=numeric,
sorting=nyt
]{biblatex}
\usepackage{hyperref}
\usepackage{xurl}
\hypersetup{breaklinks=true}

\usepackage{geometry}
\usepackage{algorithm}
\usepackage[noend]{algpseudocode}

\RequirePackage{amsmath}
\RequirePackage{amsthm}
\usepackage{amssymb}
\usepackage{comment}
\usepackage{enumitem}
\usepackage{multicol}
\usepackage{booktabs}

\newtheorem{theorem}{Theorem}[section] 
\newtheorem{lemma}[theorem]{Lemma}

\newtheorem{remark}[theorem]{Remark}
\newtheorem{definition}[theorem]{Definition} 
\newtheorem{example}[theorem]{Example}

\date{University of Helsinki}
\author{Matilda Häggblom}

\begin{document}
\title{Axiomatization of approximate exclusion}
\maketitle             
\noindent\textbf{Abstract.}
We define and axiomatize approximate exclusion atoms in the team semantic setting. A team is a set of assignments, which can be seen as a mathematical model of a uni-relational database, and we say that an approximate exclusion atom is satisfied in a team if the corresponding usual exclusion atom is satisfied in a large enough subteam. We consider the implication problem for approximate exclusion atoms and show that it is axiomatizable for consequences with a degree of approximation that is not too large. We prove the completeness theorem for usual exclusion atoms, which is currently missing from the literature, and generalize it to approximate exclusion atoms. We also provide a polynomial time algorithm for the implication problems. The results also apply to exclusion dependencies in database theory. 
\bigbreak
\noindent\textbf{Keywords.} Relational database, Exclusion dependency, Approximate exclusion, Team semantics.

\section{Introduction}

Team semantics was introduced by Hodges in \cite{hodges1997,hodges19972} and further developed by Väänänen in \cite{vaananen2007} with the introduction of dependence logic. Team semantics is suited for examining expressions about relationships between variables since formulas are evaluated in a finite set of assignments, called a \textit{team}, instead of a single assignment. Various dependencies from database theory have thus been considered in the team semantic setting. Functional dependencies are captured by dependence atoms \cite{vaananen2007}, and exclusion dependencies introduced by Casanova and Vidal in \cite{Casanova1983TowardsAS} were adapted as exclusion atoms by Galliani in \cite{galliani2012}, together with inclusion atoms.

Kivinen and Mannila defined approximate functional dependencies in \cite{KIVINEN1995129}, and Väänänen later defined approximate dependence atoms and axiomatized them in \cite{Väänänen2017}. We define, analogously, that an \textit{approximate exclusion atom} is satisfied in a team if there exists a large enough subteam that satisfies the corresponding usual exclusion atom. The definition is motivated by dependence and exclusion atoms both being downward closed, i.e., when a team satisfies an atom, so do its subteams. The approximate atoms are suitable when it is permitted that the team has some, typically small, bounded degree of error.  

We consider the following implication problem: Does a (possibly infinite) set of approximate exclusion atoms imply a given approximate exclusion atom? We show that the implication problem is axiomatizable for consequences whose approximation allows less than half of the team to be faulty, and for assumption sets for which there exists a gap between the approximations larger than the one in the consequence and the approximation in the consequence. The former restriction still allows all cases where the natural interpretation of the consequence is ''almost exclusion'', and the latter only affects infinite assumption sets.

We define a complete set of rules for exclusion atoms, whose explicit axiomatization is currently missing from the literature. Some of the rules are from the system for exclusion and inclusion combined introduced in \cite{Casanova1983TowardsAS}, with a necessary additional rule. We then generalize the system 
to approximate exclusion atoms, and prove completeness using a counterexample team that also generalizes the one for usual exclusion atoms. We also provide a polynomial time algorithm showing that the finite implication problems for (approximate) exclusion atoms are decidable. 

The results in this paper can immediately be transferred to the database setting by reading ''uni-relational database'' instead of ''team'' and ''exclusion dependency'' instead of ''exclusion atom''.

\section{Exclusion atoms}
We recall basic definitions of team semantics and exclusion atoms as defined in \cite{galliani2012}.

A team $T$ is a finite set of assignments $s:\mathcal{V}\longrightarrow M$, where $\mathcal{V}$ is a set of variables and $M$ is a set of values. We write $x_i,y_i,\dots$ for individual variables and $x,y,\dots$ for finite tuples of variables. For a variable tuple $x=\langle x_1,\dots ,x_n\rangle$, we write $s(x)$ as shorthand for $\langle s(x_1),\dots,s(x_n)\rangle$. Exclusion atoms are written as $x|y$, where $x$ and $y$ have the same length, denoted $|x|=|y|$. We recall the semantics of the usual exclusion atom:
\begin{equation*}
     T\models x|y \text{ if and only if for all } s_1,s_2\in T,\  s_1(x)\neq s_2(y). 
\end{equation*}
    
It follows from the definition that exclusion atoms are downward closed: $T\models x|y$ implies that $T^\prime\models x|y$ for all subteams $T^\prime\subseteq T$. Exclusion atoms also have the empty team property since all exclusion atoms are satisfied by the empty team. We call exclusion atoms of the form $x|x$ \textit{contradictory} since they are only satisfied in the empty team.

We write $x|y\models u|v$ if for all teams $T$, $T\models x|y$ implies $T\models u|v$. If $x|y\models u|v$ and $u|v\models x|y$, we say that the exclusion atoms $u|v$ and $x|y$ are semantically equivalent and write $x|y\equiv u|v$.

We define the rules for exclusion atoms and show that the system they form is sound and complete.

\begin{definition}\label{def_usual_exl_rules}
The rules for exclusion atoms are: 
\begin{enumerate}[align=left]
\item[(E1)] $x|x\vdash y|z$

\item[(E2)] $x|y\vdash y|x$ 

\item[(E3)] $x|y\vdash xu|yv$

\item[(E4)] $xuu|yvv\vdash xu|yv$

\item[(E5)] $xyz|uvw\vdash xzy|uwv$\hfill ($|x|=|u|$ and $|y|=|v|$) 

\item[(E6)] $xw|yw\vdash zz|xy$. 
\end{enumerate}
\end{definition}

\begin{lemma}[Soundness]
Let $\Sigma$ be a set of exclusion atoms. If $\Sigma\vdash x|y$, then $\Sigma\models x|y$.
\end{lemma}
\begin{proof}
    The proofs are straightforward. 
    For $E6$, if $T\models xw|yw$, then for every $s\in T$ there is no tuple $a$ such that $s(xy)=aa$. On the other hand, for all $s\in T$, $s(zz)=bb$ for some tuple $b$. Thus we conclude $T\models zz|xy$.
\end{proof}

Let us comment on the rules in relation to the system for exclusion and inclusion combined in \cite{Casanova1983TowardsAS}. The rules $E1$, $E2$, $E3$ and $E5$ are all included in their system and, assuming they do not distinguish between dependencies like $xuu|yvv$ and $xu|yv$, the rule $E4$ does not apply. Rule $E6$ is new and we show that it is not derivable from the rules in \cite{Casanova1983TowardsAS}. Recall the definition of inclusion atoms $Inc(x,y)$, with $|x|=|y|$, 
\begin{align*}
    T\models Inc(x,y)\text{ if and only if } &\text{for all } s\in T,\text{ there exists } s^\prime\in T\text{ such } \text{that } s(x)=s^\prime(y).
\end{align*}

\begin{remark}
   The system for inclusion and exclusion in \cite{Casanova1983TowardsAS} includes the rules $E1$, $E2$, $E3$ and $E5$ together with the rules below\footnote{In \cite{Casanova1983TowardsAS}, the rules $IE2$ and $IE3$ correspond to one rule.}, and is not complete for exclusion consequences:   Consider $x_1\neq y_1$, then $x_1w_1w_2|y_1w_1w_2$ is not contradictory. We have that $x_1w_1w_2|y_1w_1w_2\models z_1z_1|x_1y_1$ with $|x_1y_1|\neq |x_1w_1w_2|$, but no rule in this system allows such a change in arity. 
\begin{enumerate}[align=left]
\item[(IE1)] $\vdash Inc(x,x)$
\item[(IE2)] $Inc(xyz,uvw)\vdash Inc(xzy,uwv)$ \hfill ($|x|=|u|$ and $|y|=|v|$)
\item[(IE3)] $Inc(xu,yv)\vdash Inc(x,y)$
\item[(IE4)] $Inc(x,z)$, $Inc(z,y)\vdash Inc(x,y)$





\item[(IE5)] $x|x\vdash Inc(y,z)$

\item[(IE6)] $Inc(x,u)$, $Inc(y,v)$, $u|v\vdash x|y$. 
\end{enumerate}
\end{remark}

Next we define a function to identify exclusion atoms with sets of ordered pairs, such that two exclusion atoms are identified with the same set only if they are semantically equivalent. We also define sets such that for a given variable in an exclusion atom, the set contains exactly those variables that appear on the other side of the exclusion symbol in the same position as the given variable. We write $(x)_i$ to denote the projection to the $i$:th variable of $x$. Let $Var$ be a function such that $Var(x)=\{(x)_i :  1\leq i\leq|x|\}$.

\begin{definition}
    Let $|x|=|y|=n$ and $x|y$ be any exclusion atom. Define the exclusion atom's set representation $S(x|y)$ by: 
    $$S(x|y)=\{\langle(x)_1,(y)_1\rangle,\langle(x)_2,(y)_2\rangle, \dots ,\langle(x)_n,(y)_n\rangle\}.$$
For each $(x)_i$ and $(y)_i$ we define the correspondence sets  $$ \mathcal{C}_{(x)_i}=\{y_j :  \langle(x)_i,y_j\rangle\in S(x| y)\}$$ $$\mathcal{C}_{(y)_i}=\{x_j :  \langle x_j,(y)_i\rangle\in S(x| y)\}.$$
\end{definition}

We give examples of the set representation and correspondence sets for the exclusion atom $x|y=x_2y_3x_2x_4|y_1y_3y_3y_4$: \begin{align*}
S(x|y)&=\{\langle x_2,y_1\rangle ,\langle y_3,y_3\rangle ,\langle x_2,y_3\rangle , \langle x_4,y_4\rangle \}, 
\end{align*} $ \mathcal{C}_{(x)_1}=\{y_1,y_3\}$, $ \mathcal{C}_{(x)_2}=\{y_3\}$, and $ \mathcal{C}_{(y)_2}=\{y_3,x_2\}$.

Let us also define the \textit{end-constant form} for exclusion atoms $u|v$ as the semantically equivalent exclusion atom $u^\prime c|v^\prime c$ where $S(u|v)=S(u^\prime c|v^\prime c)$ and for all $\langle (u)_i,(v)_i \rangle\in S(u|v)$, if $(u)_i=(v)_i$, then $\langle (u)_i,(v)_i \rangle\in S(c|c)$ and $\langle (u)_i,(v)_i \rangle\not\in S(u^\prime|v^\prime)$. 

The next lemma shows how the rules correspond to exclusion atoms' set representations. By item \ref{equiv_usual_i} in Lemma \ref{equiv_usual}, together with soundness, it follows that if two exclusion atoms are identified with the same set, then they are semantically equivalent.   Item \ref{equiv_usual_ii} gives a set representation condition in a derivation where the rule $E6$ is used. 

\begin{lemma}\label{equiv_usual} Suppose that $|x|=|y|=n$ and $u\neq v$. Let $uc|vc$ be in end-constant form. 
    \begin{enumerate}[label=(\roman*)]
        \item $u| v\vdash_{\{E3, E4, E5\}} x| y$ if and only if $S(u| v)\subseteq S(x| y)$. \label{equiv_usual_i}
        
        \item Suppose that $S(uc| vc)\not\subseteq S(x| y)$ and $S(vc| uc)\not\subseteq S(x| y)$, then $uc| vc\vdash x| y$ if and only if there is 
 $d\in\{x,y\}$ such that \begin{align*}
     S(uc|vc)\subseteq& \bigcup_{1\leq i\leq n} \mathcal{C}_{(d)_i}\times  \mathcal{C}_{(d)_i}\cup \{\langle w_l,w_l\rangle :  w_l\in \mathcal{V}\}.
 \end{align*}\label{equiv_usual_ii}
    \end{enumerate}
\end{lemma}
\begin{proof}
\begin{enumerate}[label=(\roman*)]
    \item Immediate by the definition of the rules $E3$, $E4$ and $E5$.

    \item  Let $|u|=|v|=m$. By the assumptions and item \ref{equiv_usual_i}, we need to use rule $E6$ in the derivation. We note that one application of rule $E6$ is sufficient, since after that all variables in $uv$ have already been moved to the same side. We check the case when $d=x$ and the derivation is of the form $uc|vc\vdash_{\{E6\}} zz| uv\vdash_{\{E3, E4, E5\}} x| y$, the other cases are symmetrical. 
    \begin{align*}
     uc|vc&\vdash_{\{E6\}}zz| uv\vdash_{\{E3, E4, E5\}} x| y \\
   \Longleftrightarrow& S(zz| uv)\subseteq S(x| y) \hfill \text{ by  item \ref{equiv_usual_i}}\\ 
  \Longleftrightarrow&  \text{For all } 1\leq j\leq 2m, \langle(zz)_j,(uv)_j\rangle\in S(x|y) \\ 
    \Longleftrightarrow& \text{For all } 1\leq j\leq m, \langle(z)_j,(u)_j\rangle\in S(x|y) \text{ and }\langle(z)_j,(v)_j\rangle\in S(x|y)\\&\text{with } (z)_j=(x)_i  \text{ for some } 1\leq i\leq n \\ 
\Longleftrightarrow& \text{For all }  1\leq j\leq m, \langle(u)_j,(v)_j\rangle\in  \mathcal{C}_{(x)_i}\times  \mathcal{C}_{(x)_i} \text{ for some }  1\leq i\leq n \\
\Longleftrightarrow& S(u| v)\subseteq \bigcup_{1\leq i\leq n} \mathcal{C}_{(x)_i}\times  \mathcal{C}_{(x)_i}    \\
\Longleftrightarrow& S(uc| vc)\subseteq \bigcup_{1\leq i\leq n} \mathcal{C}_{(x)_i}\times  \mathcal{C}_{(x)_i}
\cup\{\langle w_l,w_l\rangle :  w_l\in \mathcal{V}\}.        
    \end{align*}    
\end{enumerate}
\end{proof}

We are now ready to prove the completeness theorem for exclusion atoms, by constructing a counterexample team.

\begin{theorem}[Completeness]\label{usual_compl}
Let $\Sigma\cup\{x|y\}$ be a set of exclusion atoms with $|x|=|y|=n$. If $\Sigma\models x| y$ then $\Sigma\vdash x| y$.
\end{theorem}

\begin{proof}
We assume $\Sigma\not\vdash x|y$ and show  $\Sigma\not\models x|y$. Let $X=Var(x)$, $Y=Var(y)$ and let $Z=\{z_1,z_2,\dots\}$ be the variables in $\Sigma$ not in $X\cup Y$. Construct a team $T=\{s_1,s_2\}$ with values from $\mathbb{N}$ as follows.\footnote{The completeness proof is written for when the variables in $\Sigma\cup\{x|y\}$ form a countable set. If not, we construct the counterexample team with the same method but with values from some uncountable set.} Define $s_1(x)=\langle a_1, \dots, a_n \rangle=s_2(y)$, such that for all $i,j\in\{1,2\dots,n\}$, $a_i=a_j$ if and only if $(x)_i=(x)_j$ or $(y)_i= (y)_j$. All other assigned values are unique values from $\mathbb{N}$.

\begin{table}[t]\centering
$\begin{array}{|c| ccc |ccc |ccc}
\hline  	
&(x)_1 & \dots & (x)_n &(y)_1 & \dots & (y)_n & z_1 & z_2 & \dots \\ 
\hline  
 s_1& \cellcolor{Graylight}a_1 &\cellcolor{Graylight}  -&\cellcolor{Graylight} a_n 
  &  
  & \dots
  & 
  & 
  &
  & \dots
  \\

 s_2& 
  & \dots
  & 
  &\cellcolor{Graylight}a_1 
  &\cellcolor{Graylight} -
  &\cellcolor{Graylight} a_n  
  & 
  &
  & \dots
  \\
\hline  
\end{array}$ \caption{Counterexample team for exclusion atoms.{\protect\footnotemark}}
\end{table}

\footnotetext{If $(x)_i=(y)_j$ for some indices $i$ and $j$, then they correspond to one column in the team with $s_1((x)_i)=a_i$ and $s_2((x)_i)=a_j$.}

Clearly, $T\not\models x|y$. Let $u|v\in\Sigma$. First note that $u|v$ with $u=v$ or $S(u|v)$ satisfying any of the cases in Lemma \ref{equiv_usual} would allow us to prove a contradiction. 

For the remaining exclusion atoms $u|v\in\Sigma$, we show $T\models u|v$ by considering the different forms of $u|v$ one by one while excluding the previous cases. First, consider when $Var(uv)\cap Z\neq\emptyset$:
If there is $\langle (u)_i ,z_j \rangle\in S(u|v)$, and $(u)_i\neq z_j$, then $(u)_i$ and $z_j$ have no common values in the team. Otherwise, all variables from $Z$ appear in the form $\langle z_j ,z_j \rangle\in S(u|v)$. Then one of the following cases must hold (or their symmetrical variant):
\begin{enumerate}[resume,label=(\alph*)]
    \item There is $\langle (x)_i, (y)_k \rangle \in S(u|v)$ with $(x)_i\neq (y)_k$. Then there is a possible shared value for $(x)_i$ and $(y)_k$ only if $i=k,(x)_i=(x)_k$ or $(y)_i=(y)_k$, at $s_1((x)_i)=s_2((y)_k)$. But $s_1(z_j)\neq s_2(z_j)$.

    \item There is $\langle (x)_i, (x)_k \rangle \in S(u|v)$ with $(x)_i\neq (x)_k$ (or similarly $\langle (y)_i, (y)_k \rangle \in S(u|v)$), then also $(y)_i\neq(y)_k$ by Lemma \ref{equiv_usual} \ref{equiv_usual_ii}, so $(x)_i$ and $(x)_k$ share no values. \label{item b} 
\end{enumerate}

Now suppose that $Var(uv)\subseteq X\cup Y$. The case when $Var(uv)\subseteq X$ (or similarly for $Y$) is similar to item \ref{item b}.

\begin{enumerate}[resume,label=(\alph*)]
        \item If $Var(u)\subseteq X$ and $Var(v)\subseteq Y$, then there is $\langle (x)_i, (y)_k \rangle \in S(u|v)$ with $(x)_i\neq (x)_k$ and $(y)_i\neq(y)_k$ (by Lemma \ref{equiv_usual} \ref{equiv_usual_i}), so $(x)_i, (y)_k$ have no common values. 

        \item If $Var(u)\subseteq X\setminus Y$ (or similarly for $Y\setminus X$) and $Var(v)\cap X$ and $Var(v)\cap Y$ are both nonempty, then there are $\langle (x)_i, (x)_k \rangle,  \langle (x)_l, (y)_m \rangle\in S(u|v)$ such that $(x)_k\not\in Y$ and $(y)_m\not\in X$. The only possible shared value for $(x)_l, (y)_m$ is at $s_1((x)_l)=s_2((y)_m)$, but $s_1((x)_i)\neq s_2((x)_k)$.

    \item If the intersections $Var(v)\cap X$, $Var(v)\cap Y$, $Var(u)\cap Y$, and $Var(v)\cap Y$ are all nonempty, either there are $\langle (x)_i, (y)_k \rangle, \langle (y)_l, (x)_m \rangle\in S(u|v)$ with $(x)_i\not\in Y$ and $(y)_k \not\in X$, and the only possible shared value for $(x)_i$ and $(y)_k$ is at $s_1((x)_i)=s_2((y)_k)$, but $s_1((y)_l)\neq s_2((x)_m)$. Or, there are $\langle (x)_i, (x)_k \rangle, \langle (y)_l, (y)_m \rangle \in S(u|v)$ where the only possible shared value for $(x)_i$ and $(x)_k$ is at $s_1((x)_i)=s_1((x)_k)$, but $s_1((y)_l)\neq s_1((y)_m)$.  
\end{enumerate}
\end{proof}

In the completeness proof, we construct a team with values from $\mathbb{N}$, and we note that if the set of variables in $\Sigma$ is finite, $|Z|=m$ and $|x|=|y|=n$, then we can construct a counterexample team given a model with a domain of size at least $3n+2m$.

We show in Theorem \ref{appr_dec} that the implication problem for finite sets of approximate exclusion atoms is decidable by constructing a polynomial time algorithm, from which the corresponding result follows for usual exclusion atoms. 

\begin{theorem}[Decidability]\label{dec}

Let $\Sigma\cup\{x|y\}$ be a finite set of exclusion atoms. The implication problem for whether $\Sigma\vdash x| y$ is decidable.
\end{theorem}

\section{Approximate exclusion atoms}

We define approximate exclusion atoms in an analogous way to the approximate dependence atoms in \cite{Väänänen2017}, motivated by both atoms being downward closed. 

\begin{definition}
Let $p$ be a real number such that $0 \leq p \leq 1$. $T\models x|_p y$ if and only if there is a subteam $T^\prime\subseteq T$, $|T^\prime| \leq p \cdot |T|$, such that $T \setminus T^\prime\models x| y$. 
\end{definition}
Thus $p=0$ coincides with usual exclusion, while an approximate exclusion atom with $p=1$ is always satisfied by the empty team property. For small approximations $p$, a team satisfying the approximate exclusion atom $x|_p y$ corresponds to the usual exclusion atom $x| y$ \textit{almost} being satisfied in the team. Thus the statement ''almost all individuals in last year's top $50$ ranking are not in this year's top $50$ ranking'' can be formalized as, e.g., $x_1|_{\frac{3}{50}} x_2$, meaning that for a team  $T=\{s_1,s_2,\dots,s_{50}\}$ such that for all $i\in\{1,2,\dots,50\}$, $s_i(x_1)$ is the name of the individual in last year's $i$:th place and $s_i(x_2)$ is the name of the individual in this year's $i$:th place, $T\models x_1|_{\frac{3}{50}} x_2$ holds if and only if at most three individuals remained in this year's top $50$ ranking.

We extend the definition of \textit{contradictory} atoms to include approximate exclusion atoms of the form $x|_p x$, $p<1$, since they too are only satisfied by the empty team.

As for the approximate dependence atoms in \cite{Väänänen2017}, locality fails for approximate exclusion atoms. We show this by an example, also illustrating that the approximate version of rule $E3$ in Definition \ref{def_usual_exl_rules} is not sound under a definition where teams are restricted to the variables in the atom they evaluate. 

\begin{example} \label{excl_weak_fail_restr}
Let the team $T\restriction_{xy}$ and $T\restriction_{xuyv}$ be as illustrated in Tables \ref{t2} and \ref{t3}, then  $T\restriction_{xy}\models x|_{\frac{1}{2}}y$ but $T\restriction_{xuyv}\not\models xu|_{\frac{1}{2}}yv$.

\begin{table}[ht]
\parbox{.5\linewidth}{
\centering
$\begin{array}{|cc|}
\hline	 	
x &  y   \\ 
\hline  
0 & 0  \\

1  & 2  \\
\hline 
\end{array}$
\caption{$T\restriction_{xy}$}\label{t2}
}
\hfill
\parbox{.5\linewidth}{
\centering
$\begin{array}
{|cccc|}
\hline 	
x & u & y & v  \\ 
\hline 
0 & 1 & 0 & 1 \\

0 & 2 & 0 & 2  \\

1 & 2 & 2 & 1  \\
\hline 
\end{array}$
\caption{$T\restriction_{xuyv}$}\label{t3}
} 
\end{table}

\end{example}

We define the rules for approximate exclusion atoms and show that they are sound.

\begin{definition}
    The rules for approximate exclusion atoms are: 

\begin{enumerate}[align=left]
\item[(A1)]$x|_p x\vdash y|_0 z$, for $p<1$ 

\item[(A2)] $x|_p y\vdash y|_p x$

\item[(A3)] 
$x|_p y\vdash xu|_p yv$

\item[(A4)] $xuu|_p yvv\vdash xu|_p yv$

\item[(A5)] 
$xyz|_p uvw\vdash xzy|_p uwv$\hfill ($|x|=|u|$ and $|y|=|v|$)  

\item[(A6)] $xw|_p yw\vdash zz|_p xy$ 

\item[(A7)]  $x|_q y \vdash x|_p y$, for $q\leq p\leq 1$

\item[(A8)] $\vdash x|_1 y$.

\end{enumerate}
\end{definition} 

\begin{lemma}[Soundness]
    The rules $A1-A8$ are sound.
\end{lemma}

\begin{proof}
 The rules $A1-A6$ are the approximate versions of $E1-E6$ in Definition \ref{def_usual_exl_rules}, and soundness follows similarly. For $A7$, if already $T\models x|_p y$, then for $q\geq p$, $T\models x|_q y$. For $A8$, $T\models x|_1 y$ always holds by the empty team property.
\end{proof}

We extend the definition of the \textit{set representation} to approximate exclusion atoms by defining $S(x|_p y)=S(x| y)$. Now two approximate exclusion atoms with the same degree of approximation have the same set representation only if they are semantically equivalent. We also extend the definition of the \textit{end-constant form} to approximate exclusion atoms $u|_q v$ as $u^\prime c|_q v^\prime c$ analogously.

\begin{lemma}\label{equiv_appr} Suppose that $|x|=|y|=n$, $u\neq v$ and $q\leq p<\frac{1}{2}$. Let $uc|_q vc$ be in end-constant form.
    \begin{enumerate}[label=(\roman*)]
        \item $u|_q v\vdash_{\{A3, A4, A5, A7\}} x|_p y$ if and only if $S(u|_q v)\subseteq S(x|_p y)$.\label{equiv_appr_i}
        
        \item Suppose that $S(uc|_q vc)\not\subseteq S(x|_p y)$ and $S(vc|_q uc)\not\subseteq S(x|_p y)$, then $uc|_q vc\vdash x|_p y$ if and only if there is 
 $d\in\{x,y\}$ such that \begin{align*}
     S(uc|_q vc)\subseteq& \bigcup_{1\leq i\leq n}\mathcal{C}_{(d)_i}\times \mathcal{C}_{(d)_i}\cup\{\langle z_l,z_l\rangle  :  z_l\in \mathcal{V}\}.
 \end{align*}\label{equiv_appr_ii}
    \end{enumerate}
\end{lemma}
\begin{proof}
\begin{enumerate}[label=(\roman*)]
    \item Immediate by the definitions of the rules $A3$, $A4$, $A5$ and $A7$.

    \item Analogous to the proof of Lemma \ref{equiv_usual}  \ref{equiv_usual_ii}. 
\end{enumerate}
\end{proof}

Before we prove the completeness theorem, let us remark that all sets of non-contradictory approximate exclusion atoms are satisfiable.

\begin{remark} \label{canonical_team}
   Note that any set $\Sigma$ of non-contradictory approximate exclusion atoms is satisfied by a unary team where all variables obtain values that occur only once in the team. See Table \ref{table_lemma} for an example. 
    
   \begin{table}[h]\centering $\begin{array}{|ccccc}
    \hline
           z_1  &  z_2 & z_3& z_4 & \dots \\
           \hline
       1  &  2 & 3& 4 & \dots \\
       \hline
    \end{array}$\caption{The canonical team in Remark \ref{canonical_team}.}\label{table_lemma}\end{table}
\end{remark}

Next, we generalize the counterexample team in Theorem \ref{usual_compl} and show completeness for consequences with approximations $p<\frac{1}{2}$, thus the rule $A8$ can be omitted from the system.

\begin{theorem}[Completeness]\label{excl.compl}
Let $\Sigma\cup\{x|_p y\}$ be a set of approximate exclusion atoms with $|x|=|y|=n$, $0\leq p<\frac{1}{2}$ such that if there are $u|_q v \in\Sigma$ with $q>p$, then $r=min\{q>p : \, \,u|_q v\in\Sigma\}$ exists.\footnote{As for approximate dependence atoms in \cite{Väänänen2017}, we need to avoid infinite consequences of the form $\{x|_{\frac{1}{n}} y : n\in\mathbb{N}\}\models x| y$, since for a recursive assumption set the decidability of the logical consequence would allow us to encode the halting problem.} If $\Sigma\models x|_p y$ then $\Sigma\vdash x|_p y$.
\end{theorem}
\begin{proof}

Let $0\leq p<\frac{1}{2}$ and assume that $\Sigma\not\vdash x|_p y$. We show that $\Sigma\not\models x|_p y$. Construct a team $T$ of size $k$ such that $p< \frac{l}{k}\leq r$, where $r=min\{q>p : \, \,u|_q v\in\Sigma\}$, for some positive integer $l$.\footnote{If no approximate exclusion atom in $\Sigma$ has approximation $q>p$, then $l=1$ and $k=2$ suffices.} Construct the team $T$ with values from from $\mathbb{N}$ occuring only once, except for $s_1(x)=s_{l+1}(y),$ $\dots,$ $s_{l}(x)=s_{2l}(y)$, i.e., to satisfy $x|y$ we would have to remove at least $l$ lines, and for all $e\in\{1,\dots,l\}$, $e^\prime\in\{l+1,\dots,2l\}$ and $i,j\in \{1,\dots,n\}$, $s_e((x)_i)=s_{e}((x)_j)$ and $s_{e^\prime}((y)_i)=s_{e^\prime}((y)_j)$ if and only if $(x)_i=(x)_j$ or $(y)_i=(y)_j$.

\begin{table}[t]\centering
$\begin{array}{|c| ccc| ccc |ccc}
\hline	
&(x)_1 & \dots & (x)_n & (y)_1 & \dots & (y)_n & z_1 & z_2& \dots\\ 
\hline
 s_1& \cellcolor{Graylight}a_1 &\cellcolor{Graylight}  -&\cellcolor{Graylight} a_n 
  &  
  & \dots
  &
  & 
  &
  & \dots
  \\

 s_2& \cellcolor{Graylight}b_1 &\cellcolor{Graylight}  -& \cellcolor{Graylight}b_n 
  & 
  & \dots
  &  
  & 
  &
  & \dots
  \\
 
 s_3&  
  & \dots
  & 
  & \cellcolor{Graylight}a_1 & \cellcolor{Graylight} -&\cellcolor{Graylight} a_n
  & 
  &
  & \dots
  \\

 s_4& 
  & \dots
  &
  &\cellcolor{Graylight}b_1  
  & \cellcolor{Graylight}-
  & \cellcolor{Graylight}b_n  
  & 
  &
  & \dots
  \\

 s_5&  &  \dots& 
  &  
  & \dots
  & 
  & 
  &
  & \dots
  \\
\hline
\end{array}$  \caption{Counterexample team for approximate exclusion atoms where $\Sigma\cup\{x|_p y\}$ and $r=min\{q>p : u|_q v\in\Sigma\}$ are such that  $p<\frac{2}{5}\leq r$.} 
\end{table}
Clearly, $T\not\models x|_p y$. We show that $T\models u|_q v$ for all approximate exclusion atoms $u|_q v\in\Sigma$. Let $u|_q v\in\Sigma$. If $u=v$ and $0\leq q<1$, then we derive a contradiction. If $S(u|_q v)$ satisfies any of the cases in Lemma \ref{equiv_appr} and $q\leq p$, then we can derive a contradiction, and if $q>p$ then we are allowed to remove at least $l$ lines, so $T\models u|_q v$. For all other cases, we can show similarly to the proof of Theorem \ref{usual_compl}, that $T\models u|_0 v$, from which $T\models u|_q v$ follows. 
\end{proof}

For a counterexample team constructed according to the ratio $\frac{l}{k}$, if the set of variables in $\Sigma\cup \{x|_p y\}$ is finite, $|x|=|y|=n$ and there are $m$ many variables in $Z$, then we can construct a counterexample team given a model with a domain of size at least $3ln+2lm+(k-2l)(2n+m)$.

As a direct consequence of the form of the complete proof system for approximate exclusion atoms, we obtain consequence compactness in a very strong sense: Let $\Sigma\cup\{x|_p y\}$ be as in Theorem \ref{excl.compl}, if $\Sigma\models x|_p v$, then there is some $u|_q v\in\Sigma$ such that $u|_q v\models x|_p y$.

We end the section by defining Algorithm \ref{alg}, showing that the finite implication problem for approximate exclusion atoms is decidable in polynomial time.

\begin{algorithm}[t]
\caption{$\Sigma\vdash x|_p y$?}\label{alg}
\begin{algorithmic}[1] 
\Require Finite set $\Sigma\cup\{x|_p y\}$ of approximate exclusion atoms with $|x|=|y|=n$ and \hspace*{.63cm}  $0\leq p<\frac{1}{2}$
\Ensure TRUE if $\Sigma\vdash x|_p y$, FALSE otherwise

\If{$x|_q y\in \Sigma$ or $y|_q x\in \Sigma$ with $q\leq p$} \Return TRUE \label{by a2}\EndIf
\If{there exists $u|_q v\in \Sigma$ with $u=v$ and $q<1$} \Return TRUE\label{by a1}\EndIf

 \If{$x=y$} \Return FALSE \label{false by a1}\EndIf

  \State $\mathcal{V}:=\{w_l :  w_l\text{ occurs in  }\Sigma\cup\{x|_p y\}\}$

 \State $S(x|_p y):=\{\langle(x)_1,(y)_1\rangle, \langle(x)_2,(y)_2\rangle,\dots ,\langle(x)_n,(y)_n\rangle\}$

 \State $ \mathcal{C}_{(x)_i}:=\{y_j :  \langle(x)_i,y_j\rangle\in S(x|_p y)\}$ for all $1\leq i\leq n$

 \State $ \mathcal{C}_{(y)_i}:=\{x_j :  \langle x_j,(y)_i\rangle\in S(x|_p y)\}$ for all $1\leq i\leq n$

\For{$u|_q v\in \Sigma$ with $q\leq p$ \label{loop}} 

 \State  $S(u|_q v):=\{\langle(u)_1,(v)_1\rangle, \langle(u)_2,(v)_2\rangle \dots ,\langle(u)_n,(v)_n\rangle\}$

 \If{$S(u|_q v)\subseteq S(x|_p y)$} \Return  TRUE\label{by all} \EndIf

  \If{$S(v|_q u)\subseteq S(x|_p y)$} \Return TRUE \label{by all and sym}\EndIf

 \If{
 $S(u|_q v)\subseteq \bigcup_{1\leq i\leq n} \mathcal{C}_{(d)_i}\times  \mathcal{C}_{(d)_i}\cup \{\langle w_l,w_l\rangle : w_l\in \mathcal{V}\} $ for some $d\in\{x,y\}$} \hspace*{.46cm} \Return TRUE \label{by a6}\EndIf

\EndFor
 \State \Return FALSE \label{by lemma}

\end{algorithmic}
\end{algorithm}


\begin{theorem}[Decidability]\label{appr_dec}

Let $\Sigma\cup\{x|_p y\}$, $p<\frac{1}{2}$, be a finite set of approximate exclusion atoms with $|x|=|y|=n$. The implication problem for whether $\Sigma\vdash x|_p y$ is decidable. 
\end{theorem}

\begin{proof}
We show that Algorithm \ref{alg} is sound and complete. First note that the algorithm halts for any input $\Sigma\cup\{x|_p y\}$, since $\Sigma$ is finite. 

If the algorithm returns TRUE, then $\Sigma\vdash x|_p y$ follows by the rules $A1-A7$ and Lemma \ref{equiv_appr}.

If the algorithm returns FALSE at  
step \ref{false by a1}, then $x|_p y$ is contradictory and there is no contradictory atom in $\Sigma$, so by 
Remark \ref{canonical_team} and Theorem \ref{excl.compl}, $\Sigma\not\vdash x|_p y$.
If the algorithm returns FALSE at step \ref{by lemma}, then the conditions of steps \ref{by all}, \ref{by all and sym} and \ref{by a6} are not satisfied for any atom in $\Sigma$. So by Lemma \ref{equiv_appr} \ref{equiv_appr_i} we can not derive $x|_p y$ without rule $A6$. Also, by Lemma \ref{equiv_appr} \ref{equiv_appr_ii}, the derivation does not allow the use of rule $A6$, so $\Sigma\not\vdash x|_p y$.\end{proof} 

\section{Conclusion and directions for further research}

We axiomatized the unrestricted implication problem for exclusion atoms and showed that the finite implication problem is decidable. By extending first-order logic with exclusion atoms, we obtain a logic that is expressively equivalent to dependence logic, which has partial axiomatizations \cite{kontinen_axiomatizing_2013,yang_negation_2019}, and one could consider doing the same for first-order logic with exclusion atoms. 

In the propositional setting, the expressive power of propositional logic with exclusion atoms is not yet known, and its axiomatization is also missing. For finite sets of (approximate) exclusion atoms, the completeness proofs in this paper can be encoded using propositional variables, by simply introducing enough propositional variables to encode all values in the counterexample teams in binary. However, we note the following fact about propositional exclusion atoms, rendering them more expressive than their first-order counterparts:   $$T\models p_1|p_2 \text{ if and only if } T\models =\hspace*{-.1cm}(p_1)\land =\hspace*{-.1cm}(p_2)\land p_1\neq p_2,$$ 
where $=\hspace*{-.1cm}(q_1)$ is the unary constancy atom defined by $T\models =\hspace*{-.1cm}(q_1)$ if and only if for all $s_1,s_2\in T$, $s_1(q_1)=s_2(q_1)$.

We defined and axiomatized approximate exclusion atoms for consequences with an approximation $p<\frac{1}{2}$, and with a restriction on the approximations in infinite assumption sets. 
The definition is analogous to the one for approximate dependence atoms in \cite{Väänänen2017}, and one could consider defining and axiomatizing approximate versions of other downward closed dependencies such as the degenerated dependency in \cite{Thalheim}. The approximate versions of non-downward closed dependencies, like inclusion, likely have to be captured with a different definition.

\section*{Acknowledgments}
The author was supported by the Vilho, Yrjö and Kalle Väisälä Foundation, Academy of Finland projects 336283 and 345634, and the Doctoral Programme in Mathematics and Statistics at the University of Helsinki. The results are part of an ongoing project with Åsa Hirvonen, whose advice regarding this paper has been very useful. The author also wants to thank Fan Yang and Miika Hannula for helpful discussions. The initial idea for the formulation of the algorithm was by Minna Hirvonen, who also gave suggestions regarding notation.

\printbibliography

\end{document}